\documentclass{sig-alternate-per-modified}
\usepackage[utf8]{inputenc}
\usepackage{amsmath}
\usepackage{amssymb}
\usepackage{graphicx}
\usepackage{cleveref}

\usepackage{amsmath}
\usepackage{amssymb}
\usepackage{gensymb}  

\usepackage{theorem}
\newtheorem{remark}{Remark}
\newtheorem{prop}{Proposition}
\newtheorem{corollary}{Corollary}

\newtheorem{lemma}{Lemma}

\usepackage{tikz,pgfplots}
\pgfplotsset{compat=1.18} 
\newcommand{\E}[1]{\mathrm{E}[ #1 ]}
\newcommand{\bE}[1]{\mathrm{E}\left[ #1 \right]}

\newcommand{\PP}[1]{\mathrm{P}\{ #1 \}}
\newcommand{\V}[1]{\mathrm{V}[ #1 ]}

\definecolor{greenwater}{rgb}{0.35,0.85,0.78}

\newcommand{\bu}{\mathbf{u}}   
\newcommand{\br}{\mathbf{r}}   
\newcommand{\Cr}{c}            
\newcommand{\sn}{n}            
\newcommand{\vf}{v}            
\newcommand{\Tf}{T}            

\begin{document}

    
\title{Size-Aware Dispatching to Fluid Queues}

\numberofauthors{3}
\author{
\alignauthor Runhan Xie\\
       \affaddr{University of California, Berkeley} \\
       \email{runhan\_xie@berkeley.edu}
\and
\alignauthor Esa Hyytiä \\
        \affaddr{University of Iceland} \\
        \email{esa@hi.is}
\and
\alignauthor Rhonda Righter \\
       \affaddr{University of California, Berkeley} \\
       \email{rrighter@berkeley.edu}
}


\maketitle

\begin{abstract}
    We develop a fluid-flow model for routing problems, where fluid consists of different size particles and the task is to route the incoming fluid to $\sn$ parallel servers using the size information in order to minimize the mean latency.

    The problem corresponds to the dispatching problem of (discrete) jobs arriving according to a stochastic process. In the fluid model the problem
    reduces to finding an optimal path to empty the system
    in $\sn$-dimensional space. We use 
    the calculus of variation to characterize the structure of optimal policies.
    Numerical examples shed further light on the fluid routing problem and the optimal control of large distributed service systems.
\end{abstract}

\section{Introduction}

The problem of routing jobs to parallel FCFS (first-come first-served) servers based on job sizes and server workloads in order to minimize, e.g.,
mean latency is important for many applications \cite{winston-applied-1977,johri-ejor-1989,ephremides-tac-1980}.
Though SRPT scheduling for each server is optimal in this context, FCFS is the dominant scheduling policy in practice; see, e.g., \cite{Madni}, \cite{Tirmazi}.

Even under Markov assumptions with homogeneous FCFS servers, there is no simple characterization
of the optimal dispatching policy when both the size of an arriving job and the workloads at all 
servers are observed \cite{hyytia-queueing-2022}.
We therefore consider stable dispatching systems at the fluid limit, so instead of individual jobs arriving according to a point process in time,
we model the incoming jobs as a continuum, while 
keeping the notion of job sizes as a control parameter of the fluid flow process.

The value function, or total latency, of the fluid system approximates the value function of the
underlying stochastic system, but is more computationally
tractable for determining the optimal policy and evaluating heuristics. 

The main contributions of this paper are as follows:
\begin{enumerate}
\item We formulate the fluid control problem that captures the dynamics of the original dispatching system
when the number of jobs in the system is large. 
To our knowledge, we are the first to study size-aware flow control for fluid models of queueing systems.
\item 
We show that shorter jobs should be routed to shorter queues, which reduces the problem to an optimal path problem, and 
 that the optimal path is invariant under appropriate scaling, which reduces the number of dimensions required to determine the path. We also observe
 that unbalancing the workloads is beneficial, especially 
under heavy load.
\item For some paths, and for the optimal path in heavy traffic, the fluid value function matches the value function for the original stochastic model.
\item We show, for two servers, that latency decreases as job sizes become more variable in convex ordering.
\item We develop new heuristics for the original model based on the optimal fluid policy when workloads are large, but are adjusted to avoid idling for small workloads.
\end{enumerate}

The fluid control problem 
reduces to the problem of determining an optimal path to empty the system.
We give integral expressions for costs of arbitrary paths. Fixed control actions correspond to
straight line segments, for which closed-form expressions are available.

\section{Related Work}
The dispatching problem 
has been extensively studied in the queueing theory literature, and
many policies have been proposed and shown to be optimal under different assumptions on the available information (server states and job sizes). 
If the system is homogeneous and state-aware only, then 
Join-the-Shortest-Queue (JSQ) \cite{weber1978optimal,winston1977optimality} or Least-Work-Left (LWL) \cite{daley1987certain,foss1980approximation,koole1992optimality,akgun2013partial} have been shown to be optimal. 
If the system is size-aware only, then Size-Interval-Task-Assignment (SITA) \cite{crovella-sigmetrics-1998} is shown to be optimal \cite{feng-peva-2005}.
With the routing history as the only state information, Round-Robin (RR) is optimal, and combining RR and SITA can out-perform RR or SITA alone \cite{ephremides1980simple,liu1994optimality,liu1998optimal}. 

The problem is much harder
when the dispatching policy is both size- \emph{and} state-aware. 
Sequential dispatching heuristics, that use both state and job-size information, and that route short jobs to short queues, were introduced in \cite{hyytia-itnac-2022},
and a particular one, called \textsc{Dice}, was shown numerically to have excellent performance. 
Another dispatching heuristic, CARD, 
was proposed in \cite{xie-2024} and proven to be delay optimal in heavy traffic, and was shown to perform well in simulations. Our work provides theoretical support for heuristics like CARD and \textsc{Dice} that unbalance the loads and send shorter jobs to shorter queues.

Fluid models 
have been applied to study staffing problems e.g. \cite{harrison2005method, whitt2006fluid, whitt2006staffing}, scheduling problems e.g. \cite{whitt2006fluid, atar2010cmu, chen2021scheduling, zychlinski2023managing}, and dispatching problems e.g. \cite{talreja2008fluid, perry2009responding, stanojevic2010distributed, huang2019optimal, hou2020real}. These models and their controls are often state aware, i.e., they use remaining fluid in the system for decision making, but are not job-size aware.
We take a first step in studying size-aware fluid dispatching problems.

\section{Fluid Dispatching Model}
\label{sect:model}

Each of $n$ FCFS servers has service rate $1/n$, and
work comes in as a fluid flow at rate $\lambda < 1$ comprising different
size particles with known density $f(x)$, CDF $F(x)$, and CCDF $\bar{F}(x)=1-F(x)$,
and with mean $\E{X}=1$, so that $\rho=\lambda\, \E{X} =\lambda < 1$. 
The incoming flow is to be split among servers as a function of job sizes (or size of particles).

The state information, assumed known,
is 
denoted by $\bu=(u_1,\ldots,u_\sn)$, where $u_i$ defines the backlog, or remaining workload, in queue $i$.
Thus, with no arrivals, queue $i$ would empty at time $u_i \sn$. 
A control policy (job dispatcher) $\pi$ splits the incoming flow of jobs $\lambda$ into $\sn$ sub-flows, based on job sizes and the current backlog. That is, the policy chooses arbitrary
subsets of job sizes from non-negative reals,
$A_{i}(\bu)$, 
and routes jobs of size $x\in A_i(\bu)$ to server $i$. 
We will show that for an optimal policy, each subset comprises a single non-overlapping interval, and moreover, shorter jobs Should be routed to servers with smaller backlog. 
%

The instantaneous rate of jobs, $\lambda_i(\bu)$, and the work flow, $\rho_i(\bu)$, routed to server $i$ is then 
$$
\lambda _{i}(\bu):=\lambda \int_{A_i (\bu)}f(x)\,dx \text{;}\,\,\rho_{i}(\bu):= \lambda \int_{A_i (\bu)}xf(x)\,dx
$$
where $\sum \lambda_i(\bu)=\lambda =\rho =\sum \rho_i (\bu)$. For ease of notation, 
we often omit ``$(\bu)$'' and 
write $\lambda_i$ and $\rho_i$, but 
we remind the reader that these rates are constantly changing according to the control policy.

Given our objective of minimizing latency, each arriving job assigned to server $i$ incurs cost $u_i \sn$, so
server $i$ incurs waiting time cost at rate
$\lambda_i u_i \sn$, 
and the total cost rate 
is
\begin{equation}\label{eq:cost-rate-Fu}
\Cr(\bu) := n\sum_i \lambda_i u_i.
\end{equation}
For all meaningful policies
$\bu = 0$ is an absorbing state where no further costs are incurred, but there may be others.
Let $\sn_0$ denote the number of empty servers, 
$
\sn_0 := ||\{ i : u_i=0\}|| = \sum_i \mathbf{1}(u_i=0).$
Whenever $\sn_0 \cdot (1/n) \ge  \rho$, the empty servers
can process all incoming jobs without increasing their backlogs, incurring no further costs, so we call any state with $\sn_0 \ge  n\rho$ an absorbing state. We can assume, without loss of generality, that the system stops (is absorbed) in such a state. The question is how to move from the current state $\bu$
to an absorbing state at minimal cost. It is easy to show that before absorption, incoming fluid should be routed to servers with zero fluid at rate $\rho_i=1/n$, and all servers should work at rate $1/n$.

When fewer than $\sn^*$ servers are empty, we have
\begin{equation*}
  \label{eq:u-dot}
  \dot{u}_i = \rho_i - \frac{1}{\sn} =: - \nu_i,
\end{equation*}
where $\nu_i$ is queue $i$'s instantaneous drainage rate, 
with $\rho_i=1/n$ if $u_i=0$.
The cost (value function) of 
policy $\pi = (A_1(\bu),...,A_n(\bu))$, with $\tau_{\pi}(\bu)$ denoting the time to absorption, is
\begin{equation*}\label{eq:costs-vf}
v_{\pi}(\bu) := \int _0^{\infty}c(\bu_{\pi}(t))\,dt
= \int _0^{\tau_{\pi}(\bu)} c(\bu_{\pi}(t))\,dt \forall \bu.
\end{equation*}

\begin{remark}[Random Split]
\label{rmk:rnd}

With load-balancing random split (RND),
the total cost incurred is 
$$
C_i
=\int_0^{\frac{\sn}{1-\rho} u_i} c_i(t)\,dt
= 
n\int_0^{\frac{\sn}{1-\rho} u_i} \frac{1-\rho}{\sn}t\,dt
= \frac{\rho nu_i^2}{2(1-\rho)},
$$
\begin{equation*}\label{eqn:rnd}
v_{\text{RND}}(\bu) = \sum_{i=1}^n C_i=\frac{\rho \sn}{2(1-\rho)} \left( u_1^2 + \ldots + u_{\sn}^2\right),
\end{equation*}
which, for $\rho<1$ is finite for any job size distribution, regardless of $\E{X^2}$.
Moreover, 
this matches the value function of the stochastic system; for
each M/G/1-FCFS queue,
$v(u)=\rho u^2/(2(1-\rho))$ \cite{hyytia-ejor-2012,hyytia-peva-2014}.
\end{remark}

\section{Analysis 
with \lowercase{n} Servers}
\label{sect:general-case}

We first show that the optimal policy is a dynamic SITA-type policy, which, for a given $\bu$, 
splits job sizes into $\sn$ non-overlapping intervals, 
sending shorter jobs to servers with less work. That is, with 
$u_1 \ge u_2 \ge \ldots \ge u_{\sn}$, the optimal policy is determined by size thresholds $0=h_0 (\bu )\le\cdots\le h_{\sn} ( \bu )=\infty$, 
so that $A_{\sn+1-i}(\bu) = [h_{i-1}(\bu),h_i(\bu))$. 

\begin{prop}\label{prop:short-to-short}
  The optimal policy assigns the shortest jobs to the server with the smallest
  backlog, the next size interval to the server with the second smallest backlog, and so on.
\end{prop}
\begin{proof}
We prove this by contradiction. Suppose $u_i > u_j$ and the optimal policy routes jobs of size $(x, x+\delta)$ to server $i$ and jobs of size $(y, y+\delta')$ to server $j$, where $x+\delta < y$. We choose $\delta$ and $\delta'$ so that 
\[
    \E{X\cdot\mathcal{I}(x\leq X\leq x+\delta)} = \E{X\cdot\mathcal{I}(y\leq X\leq y+\delta')},
\]
i.e., the loads contributed by the two streams of jobs are equal. Let $\lambda_x=\lambda P(x\leq X\leq x+\delta)$ and $\lambda_y=\lambda P(y\leq X\leq y+\delta')$, so $\eta := \lambda_x - \lambda_y>0$. 
The cost rate due to these subflows of jobs is 
\[
    c_{ij} = n(\lambda_xu_i + \lambda_yu_j).
\]
Interchanging the corresponding job flows reduces the cost rate:
\[
    c_{ji} = n(\lambda_xu_i + \lambda_yu_j) = n(\lambda_yu_i + \lambda_x u_j +\eta (u_i - u_j)) < c_{ij},
\]
which contradicts the optimality of the proposed policy.
\end{proof}
\begin{corollary}[Structure of optimal policies] {}
The optimal dispatching policy can be characterized by $\sn-1$
appropriately chosen {\em state-dependent} thresholds $h_i(\bu)$.
Therefore, the optimal policy from any initial point $\bu$ is determined by a path from that point to an absorbing state.
\end{corollary}

\subsection{The Value Function in Terms of Paths}

Note that there is a one-to-one correspondence between a policy $\pi$ starting in some state $\bu$ and a path $\br(s)$ from $\bu$ to $(0,\ldots,0)$, 
where $s$ is the curve parameter.
By time reversal, we can also consider paths from the origin to $\bu$.
We shall consider only paths $\br(s)= (r_1(s), \ldots, r_{\sn}(s))$ that are 
admissible by an appropriately chosen
control policy.
Then $\br(0)=(0,\ldots,0)$, 
$\br(s_1) = \bu$, and
$\br'(s)$ defines the direction of the path at point $s$,
which must coincide with the direction of the drainage rates
$\boldsymbol{\nu} = (\nu_1,...,\nu_n)$, $\boldsymbol{\nu}(s) = A(s)\, \br'(s),$
where the scalar function $A(s)$ adjusts the ``speed'' at curve point $s$ to
the actual drift. Before the path reaches an absorbing state (while $n_0<n^*$),
the total drift is $\nu_1 + \ldots \nu_{\sn} = 1 - \rho$,
giving
$A(s)
= (1-\rho)/(r_1'(s) + \ldots + r_{\sn}'(s)),$
where $r_1'(s) + \ldots + r_{\sn}'(s) > 0$ as long as $\rho<1$. 
Once the path reaches an absorbing state, $r_i'(s)=0$ if $u_i =0$, and $r_i'(s)=1/n$ if $u_i >0$.
Then, 
\begin{equation*}\label{eq:nu-s}
\boldsymbol{\nu}(s) = \frac{1-\rho}{r_1'(s)+\ldots + r_{\sn}'(s)} \,
\begin{pmatrix}
  r_1'(s), & \ldots & ,r_{\sn}'(s)
\end{pmatrix}.
\end{equation*}
Thus, the (Euclidean) length of the drainage vector $||\boldsymbol{\nu}(s)||$ along the path $\br(s)$ is
\begin{equation*}\label{eq:nu-s2}
||\boldsymbol{\nu}(s)||
= \frac{1-\rho}{
  r_1'(s)+\ldots + r_{\sn}'(s)} || \br'(s) ||.
\end{equation*}
Note that $\nu_i(s)<0$ if server $i$ is receiving more work than it can process (at the given time).

\begin{lemma}[Value function]
The total accumulated cost
along a work-conserving path $\br(s)$
is
\begin{equation}
\label{eq:total-rs}
  \vf_{\br}(\bu)
  = \frac{1}{1-\rho} \int_0^{s_1} \Cr(s) (r_1'(s) + \ldots + r_{\sn}'(s))\,ds,
\end{equation}
where $\Cr(s)$ is the cost rate (per unit time) that depends on the respective control actions according to
\eqref{eq:cost-rate-Fu}.
\end{lemma}

As $\br'(s)$ defines 
$\lambda_i$ and 
$\rho_i$, the cost rate
$\Cr(s)$ depends solely on $\br(s)$ and $\br'(s)$,
and therefore
\eqref{eq:total-rs} can be evaluated (at least numerically if not in closed form) for
any given path $\br(s)$.

\subsection{Scale-free paths}

Let us consider a family of ``scale-free'' paths that are obtained by {\em scaling} a reference path $\br_0(s)$,
\begin{equation*}\label{eq:scaling}
    \br(s) = \alpha \br_0(s),
\end{equation*}
where $\br_0(s)$ is some fixed path to $\bu=\br_0(s_1)$ and $\alpha>0$ is a free scaling parameter. 
Substituting $\br(s)$ into \eqref{eq:total-rs}
reveals a quadratic relationship between the 
two value functions:
\begin{lemma}\label{lem:vf-scaling}
The value function of 
$\br(s)=\alpha \br_0(s)$ is
\begin{equation*}\label{eq:vf-scaling}
    \vf_{\br}(\alpha\bu) =
    \alpha^2 \vf_{\br_0}(\bu).
    \end{equation*}
\end{lemma}

Lemma~\ref{lem:vf-scaling}
gives the following two important structural properties of the optimal policy.
\begin{corollary}\label{cor:scale-free-optimal}
    Scale-free paths are optimal.
\end{corollary}
\begin{proof}
By contradiction, suppose $\mathbf{r}(s)$ and $\tilde{\mathbf{r}}(s)$ are
optimal paths for two states $\mathbf{u}$ and $\tilde{\mathbf{u%
}}=\alpha \mathbf{u}$ such that $v(\mathbf{u}) \ne \alpha^2 v(\mathbf{u}_0)$,
so $\tilde{\mathbf{%
r}}(s) \ne \alpha\,\mathbf{r}(s)$.  
Then either $\mathbf{r}%
(s)$ or $\tilde{\mathbf{r}}(s)$ cannot be optimal. 
First, if $v(\alpha\mathbf{u}) > \alpha^2 v(\mathbf{u})$, then 
$\tilde{\mathbf{r}}(s)=\alpha \mathbf{r}(s)$ has lower cost than $%
\mathbf{r}(s)$. 
The other case is similar.
   \end{proof}


\begin{corollary}
\label{cor:theta}
    The value function for any scale-free path 
    scales quadratically in the Euclidean distance $|\bu|$,
\begin{equation*}\label{eq:r2w}
    \vf_{\br}(\bu) =
    |\bu|^2 w(\theta),
    \end{equation*}
    where $w(\theta)$ is the value function at unit distance in the direction $\theta$ in $\sn$-dimensional space.
\end{corollary}

Due to the scaling property, we need only determine $w(\theta)$ in an $\sn-1$ dimensional surface 
to obtain the optimal paths for every state in the $\sn$-dimensional space.

\section{Optimal Paths with Two Servers}
\label{sect:two-server}

For $n=2$ servers, both with service
rate $1/2$, 
we assume that $u_{1}\geq u_{2}$, so that $u_1$ corresponds to the horizontal axis.
Now
it is also convenient to consider an alternate state
representation, $(x,y)$, where $x=u_{1}+u_{2}$ is the total backlog and $y=u_{1}-u_{2}$ is the imbalance. 
From the scale-free property (Corollary \ref{cor:scale-free-optimal}) 
the optimal path depends only on the relative imbalance, 
either captured by $\theta$, the angle for the point $(u_1,u_2)$, or the relative queue difference, 
$\hat{y}=y/x$. 
This means that if, for some point $(x_0,y_0)$ with relative imbalance $\hat{y}_0$, the optimal $\hat{y}'_0=0$, i.e., the relative imbalance should not change, then this will be true until the system empties. We call such a path a straight to the origin (STO) path. 
Because our paths are continuous, the path will visit a point with a given $\hat{y}$ either exactly once, or will maintain $\hat{y}$ on an STO path.
Indeed, we have the following corollary.

\begin{corollary}[Monotonicity of optimal paths]
\label{cor:monotonicity} For two servers, 
the optimal fluid path is 
monotonic in $\hat{y}$. 
That is, from any initial point $(x,y)$, the optimal path will always increase or decrease the relative imbalance until the system empties or until some $\hat{y}_0$, after which it will follow an STO path.
\end{corollary}

A consequence is that if the optimal policy empties queue 2 before queue 1, queue 2 will remain empty, even if $\rho>1/2$. Similarly, if the optimal path moves to perfectly balanced queues, they will remain balanced.

For each point $(x,y)$ any path is characterized by $y'=dy/dx$, which in turn, for an optimal path, is controlled by a
 single job-size threshold $h=h(x,y)=h(u_1,u_2)$: jobs shorter than $h$ are
assigned to queue 2 and the rest to queue 1.
Let $%
g(h)$ denote the load due to particles (jobs) smaller than $h$, 
and recalling that $\lambda =\rho $, we have
\begin{equation*}
g(h):=\rho \,\int_{0}^{h}z\,f(z)\,dz
=\rho\, \PP{ X\leq h }\cdot E[X|X\leq h].
\label{eq:rho-h}
\end{equation*}%
Then, while $u_2>0$ ($y<x$), $u_{2}$ changes at rate $g(h)-1/2$, $u_{1}$ changes at rate $\rho-g(h)-1/2$, $x$ changes at rate $\rho-1$, and 

\[
    y^{\prime }=\frac{(\rho -g(h)-1/2)-(g(h)-1/2)}{\rho -1}=\frac{2g(h)-\rho }{%
1-\rho },
\]
so $g(h)=(\rho +(1-\rho )y^{\prime })/{2}$. 
Because $0\leq g(h)\leq \rho $, we have, for $y<x$, that $|y^{\prime }|=|y^{\prime }(x)|\leq \rho /(1-\rho )$. That is, as the load increases there is more room to maneuver in terms of changing the path direction, including having the ability to increase (temporarily) the load at one server if $\rho>1/2$.

To capture total costs in terms of the control $y^{\prime }$, let ${\Phi }%
(y^{\prime })$ denote the fraction of jobs forwarded to queue 2, 
\begin{equation*}
{\Phi }(y^{\prime }):=F\left[ g^{-1}\left( \frac{\rho +(1-\rho )y^{\prime }}{%
2}\right) \right] .  \label{eq:Fyd}
\end{equation*}%
The cost rate is then
\begin{equation*}
2\rho (u_{2}\,\Phi
(y^{\prime })+u_{1}\,(1-\Phi
(y^{\prime })))=\rho (x+(1-2\Phi
(y^{\prime }))y),
\label{eq:cost-rate-Fu-2}
\end{equation*}
and the total cost for an arbitrary path $y^{\prime }=y^{\prime }(x)$ (and $y(x)$) from initial state $(x_{1},y_{1})$ is
\begin{equation*}
\label{eq: total cost}
\vf(x_{1},y_{1}) =
\frac{\rho }{1-\rho }\left[ x_{1}^{2}/2+\int_{0}^{x_{1}}(1-2\Phi
(y^{\prime }))y\,dx\right] .
\end{equation*}%
Because $%
x_{1}^{2}$ is independent of the path, an optimal policy minimizes 
\begin{equation*}
\label{eqn:T1}
\Tf=\int_{0}^{x_{1}}(1-2\Phi (y^{\prime })y)\,dx.
\end{equation*}%

As observed in Corollary~\ref{cor:monotonicity}, the optimal $\hat{y}$
will be monotone until some $\hat{y}_0 = y_0/x_0$, after which it will be constant until $(x,y)=(0,0)$. For such an STO path, $y'=\hat{y}_0$ and
\begin{equation*}
\label{eqn:STO}
T_{\text{STO}}(x_0,y_0)
=
x_0^2\hat{y}_0(1-2\Phi (\hat{y}_0)).
\end{equation*}
When $\hat{y}=1$ ($u_2=0$), $%
\Phi (\hat{y})=\Phi (1)=F(g^{-1}(1/2))$, and
$T_{\text{STO}}(x_{0},x_{0})=
x_1^2(1-2\Phi (1))$. 
If $u_2=0$ and $\rho <1/2$, then $\Phi (1)=1$ and $T_{\text{STO}}(x_{0},x_{0})=0$. \ 

Because $F$ and $g$ are nondecreasing, so is $\Phi $, and $\Phi (%
\hat{y})\geq \Phi (0)>1/2$. Taking the derivative we have the following.
\begin{prop} The cost of an STO path, from any point $(x_0,y_0)$ 
directly to the origin, is decreasing in 
$\hat{y_0}=y_0/x_0$, for any job-size distribution.
\end{prop}

We conjecture that the optimal $\hat{y}'\ge 0$ (the relative imbalance should increase for the optimal path).
This is supported numerically as well as by the proposition above.


Let us now consider the optimal policy when $\rho \rightarrow 1$. Define the MWL (most work left) policy as the policy that routes all jobs (fluid) to the long queue until the short queue is empty; then short jobs are routed to the short queue such that the rate of fluid to the short queue is $\max \{\rho ,1/2\}$.
It is not hard to derive the following, where $\bar{\Phi}(x)=1-\Phi (x)$.
\begin{prop}
\label{prop:MWL}
     The MWL value function, $\vf_{\text{MWL}}(\bu)$, is
     \begin{equation*}\label{eq:vf-rho-1}
    2\rho[2 u_1 u_2 + (2\rho-1) u_2^2 +
         \frac{\rho \bar{\Phi} (1)}{1-\rho}(u_1+u_2(2\rho-1))^2]
         \end{equation*}
         and MWL is asymptotically optimal, with
         \begin{equation*}
         (1-\rho )v_{\text{MWL}}(\mathbf{u})\rightarrow 
2\rho \bar{\Phi} (1)(u_{1}+u_{2})^{2}=2\rho \bar{\Phi} (1)x^2.
\end{equation*}
\end{prop}
The first two terms in the MWL value function correspond to the cost to empty the short queue, and the last term (the dominating term in heavy traffic) is the cost for the STO path that then empties the long queue once the short queue is empty.
For the original stochastic dispatching system with identical servers, the mean waiting time is
$\E{W}=\E{ \vf(X,0) }$ \cite{hyytia-peva-2024}, i.e.,
$
(1-\rho) \E{W}\rightarrow 2\rho (1-\Phi (1))\bE{  X^2 }$.
This agrees with the exact analysis of the asymptotically optimal policy for the original
stochastic system \cite{xie-2024}, suggesting that the far more tractable fluid value function may be a good approximation for the 
original value function 
even when $\rho <1$. 


We now consider the effect of job size variability on the dispatching cost. Intuitively, since the policy uses job-size information, we expect more variability to decrease costs, which is in fact the case. Recall that for two random variables $X$ and $Y$, $X$ is more variable than $Y$ in the convex sense, $X\geq _{cx}Y$, if $Ef(X)\geq Ef(Y)$ for all convex functions $f$.
The proof of the following is in the appendix.
\begin{prop}
\label{prop:cx}
$X\geq _{cx}Y\Rightarrow $ the cost is lower for any path when $X$ is a random job size rather than $Y$, so the
optimal cost is also lower.
\end{prop}

Now let us consider the extreme case, where all jobs have the same size, $X \equiv 1$. This is equivalent to allowing an arbitrary job-size distribution but job sizes are not observed or are not used in the dispatching algorithm. In this case, we can show that
any work-conserving dispatching policy is optimal: the proof is in the appendix.

\begin{prop}
\label{prop:size-unaware}
  All size-unaware work-conserving dispatching policies are equally good.
\end{prop}

\section{Numerical Examples}
\label{sect:examples}
\subsection{Optimal fluid paths}
Let us first study how the shape of the job size distribution
affects the optimal trajectories. Figure~\ref{fig:varying-sizes} shows the optimal paths for exponential, uniform, and 
both bounded and normal Pareto distributions. 
The bounded Pareto distribution with
$\alpha=1$ is truncated 
to $[1/66,6]$, so that the mean
is approximately one, and the variance is approximately 2.
The variance for the Pareto is infinite.
%

The scaling property is clearly visible, and the difference in job-size distributions matters more when the relative imbalance is larger.
It also seems advantageous to empty one queue at maximal speed when backlogs are identical
regardless of the job-size distribution. (The feasibility conditions are indicated by dashed lines on the graph.)
The trajectories are consistent with our conjecture that the relative imbalance always increases 
along any optimal trajectory to the origin (or until one queue empties). 

\begin{figure}
\centering
\includegraphics[width=55mm]{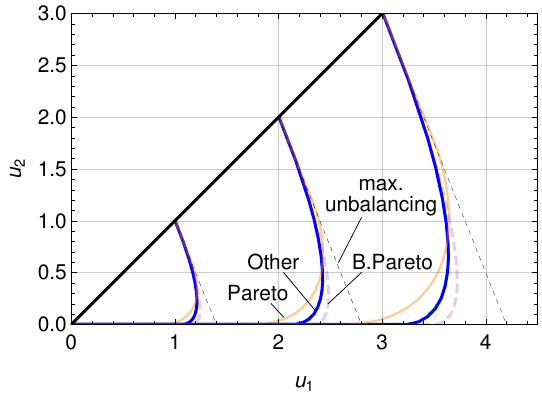}
\caption{Optimal paths with exponential, uniform, and
two Pareto distributions when $\boldsymbol{\rho}\boldsymbol{=0.7}$.
}
\label{fig:varying-sizes}
\end{figure}

\begin{figure}[tb]
\centering
\scriptsize
\includegraphics[width=55mm]{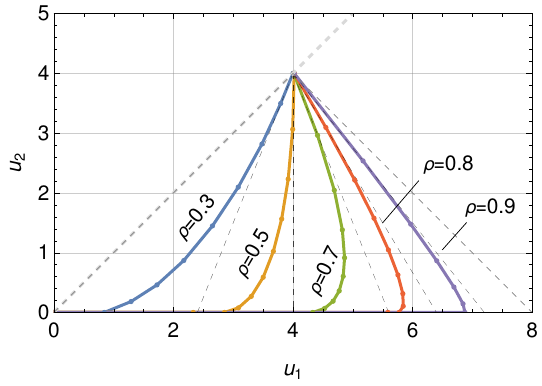}
\caption{Optimal paths depending on $\boldsymbol{\rho}$ for exponential job sizes.}
\label{fig:exp-results-rho}
\end{figure}

Figure~\ref{fig:exp-results-rho} shows
the optimal paths with exponentially distributed job sizes at different loads.
As $\rho$ increases, the
optimal strategy unbalances the queues more, consistent with our heavy-traffic result. 

\subsection{Heuristic policies}

\begin{figure}
    \scriptsize
    \begin{tabular}{cc}
    \includegraphics[width=40mm]{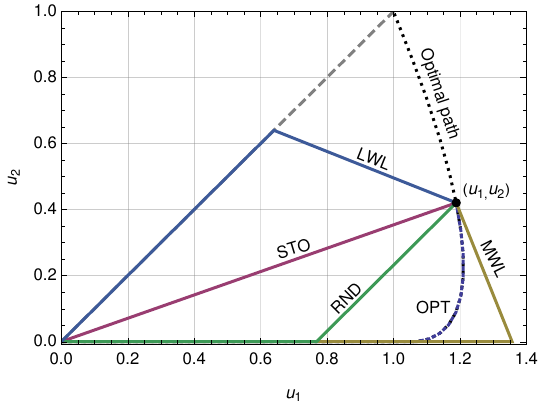} &
    \includegraphics[width=40mm]{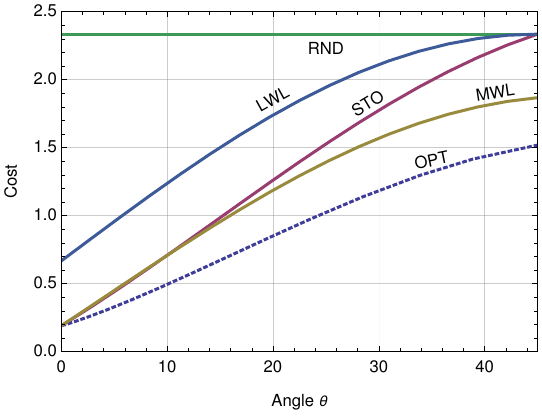}  
    %
    %
    \\
    (a)\; Paths for different policies &
    (b)\; Cost at unit distance 
    \end{tabular}
    \caption{Heuristic and optimal paths and costs
    with exponential job sizes and $\boldsymbol{\rho}\mathbf{=0.7}$.
    }
    \label{fig:fluid-ex1}
\end{figure}

Next, we compare some heuristic routing policies with the optimal dispatching policy for fluid flow with exponential job sizes and $\rho=0.7$. 

Figure~\ref{fig:fluid-ex1}(a) depicts the paths for different
heuristics and the optimal policy.
The OPT path unbalances the workloads
strongly, but not as much as MWL, which does it maximally.
All of the chosen heuristics and the optimal policy have the scaling property, 
so the costs for the whole $(u_1,u_2)$-plane can be characterized by following an arc at unit distance.
The resulting costs (value functions) are shown in
Figure~\ref{fig:fluid-ex1}(b)
as a function of the angle $\theta$ for points at unit distance, $||\bu||=1$. The costs for all policies are increasing
in $\theta$ (decreasing in the imbalance). When the fluid system starts with a large imbalance (small $\theta$), RND and LWL have significantly larger costs than STO or MWL.

Let us next consider heuristics for the original stochastic dispatching system, using the value function of the fluid model
as an approximation for the value function of the original system, and again assuming
exponential job sizes. 

We propose two heuristics that use the optimal fluid policy when backlogs are large, but adjusts to avoid idling when the short queue length is small.
The F-BLB (Fluid with a buffer lower bound) policy uses the fluid optimal policy when the backlog of the short queue is greater than some threshold, or buffer lower bound, $u_2>u_B$, and uses the LWL policy otherwise.
Numerically we observe that $u_B=3$ gives good results. The \textsc{F-blbh} policy uses the fluid optimal policy when $u_2>u_B$ and the arriving job size exceeds a threshold $h_S$. Jobs of size less than $h_S$ are always sent to the short queue, and LWL is followed if $u_2 \le u_B$.
In the numerical examples, we use $(u_B,h_S)=(2,1.5)$.

We also consider two sequential heuristics, \textsc{Dice} \cite{hyytia-itnac-2022} and CARD \cite{xie-2024}. \textsc{Dice} routes a job of size $x$ to the queue with the least backlog if $u_2+x<\tau$, where we use $\tau=6$, i.e., the virtual buffer
can accommodate six average-sized jobs. CARD uses two job-size thresholds to define ``small,'' ``medium,'' and ``large'' jobs, and it always routes small jobs to the short queue and large jobs to the long queue, and it routes medium jobs to the short (long) queue if the short queue length is below (above) some threshold.
\textsc{Dice} generally tends to perform better than the sequential CARD heuristic \cite{xie-2024}, while requiring a single tuning parameter that works well for all loads, rather than CARD's three parameters that all need to be tuned for each load.

Finally, we consider the Short-to-Short and Long-to-Long (\textsc{SSLL}) heuristic. 
Jobs shorter than a threshold $h_S$ are always routed to the shorter queue, and the rest go to the longer queue.
We use the load balancing threshold, $h\approx 1.678$
(that is specific to job-size distribution). This is also known as (static) SITA with switch, which is an improvement on SITA, where short jobs are always sent to a given fixed queue, regardless of the relative backlogs.


\newcommand{\policy}[1]{\item {\textsc{#1}:} }

\begin{figure}[t]
    \centering
    \includegraphics[width=64mm]{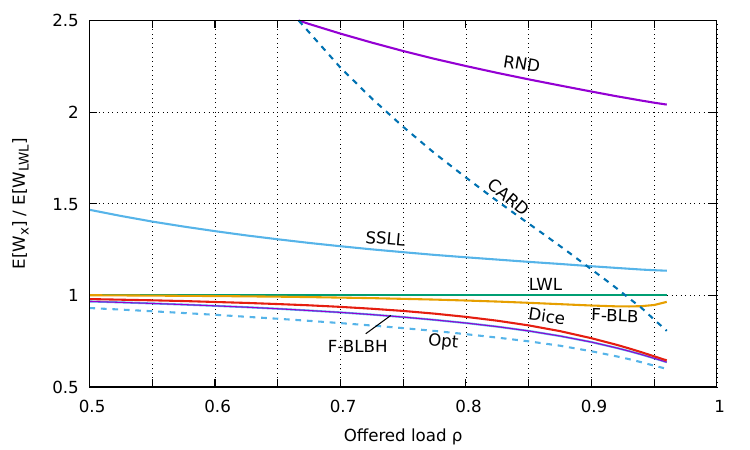}
    \caption{
    Performance with different heuristic policies 
    relative to \textsc{LWL}
    as a function of $\boldsymbol{\rho}$.}
    \label{fig:simu}
\end{figure}

Figure~\ref{fig:simu} shows the numerical results as a function of $\rho$ for the heuristics and the optimal policy.
The $y$-axis is scaled by 
the waiting time with \textsc{LWL}, as a base-line policy,
i.e.\ the relative performance metric 
for policy $\pi$ is the ratio,
$$
\frac{ \E{W_\pi} }{\E{W_{\textsc{LWL}}} }.
$$
The numerical results are based on average waiting
time observed during relatively long
simulation runs of 100-200 million arrivals.
We can observe that
\textsc{F-blb} does a better job than \textsc{LWL}
as load increases (\textsc{LWL} is optimal when
the system is lightly loaded). 
CARD, while provably optimal as the load goes to 1, does not perform well at lower loads.
On the other hand, the rather simple \textsc{Dice} and
\textsc{F-blbh} heuristics
have near-optimal performance across all loads. We also note that if the initial loads are balanced, LWL is the individually optimal policy for strategic jobs that choose a queue upon arrival, so we can view the curve for the optimal policy in Figure 4 as a measure of the ``cost of anarchy.''

\section{Future Work}
Many of our fluid results hold for two servers. We would like to extend our results and heuristics to more than two servers. One possibility for a heurisitc approach, along the lines of the CARD heuristic, is to carefully control the fluid to the short queue, while dividing the remaining fluid to balance the loads of the other queues.

Heuristics based on the fluid model work well when there are a large number of jobs in all queues. We expect that they will work well for a model in which there are occasional large bursts of traffic, and queues have finite buffers, so that overflows are forced into shorter queues. We plan to explore the impact on the optimal fluid paths when queues have finite buffers, as well as the effect of bursty arrivals.

\section{Appendix}
\subsection{Proof of Proposition \ref{prop:cx}}
Consider two job-size
distributions, $X\sim F_{X}$ and $Y\sim F_{Y}$, with similar subscripts for $T$, $g$, and $g^{-1}$. Because $F$, $g$, $g^{-1}$, and $z=z(y')$ are non-decreasing, we have
\begin{eqnarray*}
&&g_{X}(h) \leq g_{Y}(h)\quad \forall h\ge 0 
\\
&\Rightarrow& g_{X}^{-1}(z)\geq
g_{Y}^{-1}(z) \quad \forall z \in [0,\rho]
\\
&\Rightarrow& \Phi _{X}(y^{\prime })
\geq
\Phi _{Y}(y^{\prime }) \quad \forall y^{\prime } \in [-\rho/(1-\rho),\rho/(1-\rho)]
\\
&\Rightarrow& T_{X}\leq T_{Y}\text{.}
\end{eqnarray*}%
Also,
\begin{eqnarray*}
g_{X}(h) &=& \rho \int_{0}^{h}xf_{X}(x)dx= \rho E[X|X\leq h]P\{X\leq h\} \\
&=&\rho[1-E[X|X>h]P\{X>h\}] \\
&=&\rho[1-E[(X-h)^{+}]].
\end{eqnarray*}%
Because $\rho[1-(x-h)^{+}]$ is convex in $x$, the result follows. 

\subsection{Proof of Proposition \ref{prop:size-unaware}}
Here our admissible dispatching policies are the fraction of jobs, $p$, to be routed to server 2, depending on the backlogs, $(u_1,u_2)$, and by work conserving, we assume that when $u_2 =0$, the fraction of jobs routed to server 2 will be 
$p=1/(2\lambda)=1/(2\rho)$
(the maximal proportion to incur 0 costs, by making the load = 1/2).
Given the fraction of jobs routed to server 2 is $p$,
  $$
  y' = \frac{(2p-1)\rho}{1-\rho}
  \qquad
  \Rightarrow
  \qquad
  p = \frac{(1-\rho)y'+\rho}{2\rho},
  $$
  and the cost rate is given by
  $$
  \Cr = 2 \rho (p u_2 + (1-p)u_1)
  = 2\rho \left(\frac{x+y}{
  {2}} - 
  p y \right)
  = \rho x - (1-\rho)y' y
  $$
  which is similar to our earlier cost rate, with $p$ replacing $\Phi$. The total cost rate is
  
 $$
 v_{det} (x_1,y_1)=\frac{1 }{1-\rho }\left[ \rho x_{1}^{2}/2-\int_{0}^{x_{1}}(1-\rho)y'y\,dx\right].
$$
Thus, ignoring constants and $x_1^2$, which does not depend on the path, we obtain the following equivalent variational problem for the optimal path:
  $$
  \Tf = \int_0^{x_1} y' y\,dx = \max.
  $$
  Integrating $\Tf$ by parts gives
  $$
  \int_{0}^{x_1} y' y\,dx = \int_{0}^{y_1} y \,dy = \frac{y_1^2}{2},
  $$
  i.e.\ {\em any} allowed (work-conserving) path is optimal.

\bibliographystyle{IEEEtran}
\bibliography{references}

\end{document}